\documentclass[11pt]{article}

\usepackage[english]{babel}

\usepackage{graphicx}
\usepackage[colorlinks=true, allcolors=blue]{hyperref}

\usepackage[title]{appendix}

\usepackage{amssymb,amsmath,amsthm,amsfonts,amstext}
\usepackage[left=1.5cm,right=1.5cm,top=1.5cm,bottom=1.5cm]{geometry}

\usepackage{setspace}
\usepackage{enumitem}

\usepackage{enumerate}

\usepackage[mathscr]{euscript}

\usepackage{soul}

\usepackage[dvipsnames]{xcolor}

\usepackage{tikz}
\usetikzlibrary{arrows}
\usetikzlibrary{arrows.meta}
\usetikzlibrary{shapes}
\usetikzlibrary{backgrounds}
\usetikzlibrary{positioning}
\usetikzlibrary{decorations.markings}
\usetikzlibrary{patterns}
\usetikzlibrary{calc}
\usetikzlibrary{fit}
\usetikzlibrary{decorations}
\usetikzlibrary{decorations.pathreplacing}

\usepackage[framemethod=TikZ]{mdframed}

\usepackage[noend]{algpseudocode}
\makeatletter
\def\BState{\State\hskip-\ALG@thistlm}
\makeatother

\usepackage[ruled]{algorithm2e}

\newtheorem{mdalgorithm}{Algorithm}

\newenvironment{ourbox}{\begin{mdframed}[hidealllines=false,innerleftmargin=10pt,backgroundcolor=white!10,innertopmargin=2pt,innerbottommargin=5pt,roundcorner=10pt]}{\end{mdframed}}

\newtheorem{theorem}{Theorem}

\newtheorem{proposition}[theorem]{Proposition}
\newtheorem{lemma}[theorem]{Lemma}
\newtheorem{fact}{Fact}

\newtheorem*{remark}{Remark}
\newtheorem{definition}{Definition}

\newtheorem*{assumption}{Assumption}

\newcommand\IGNORE[1]{}

\newcommand{\R}{\ensuremath{\mathbb R}}

\newcommand{\opt}{\textsc{opt}}

\newcommand{\F}{\mathcal{F}} 
\newcommand{\Ls}{\mathscr{L}}

\newcommand{\T}{\mathscr{T}}

\newcommand\hE{\widehat{E}}

\newcommand\ASC{\mathrm{Cover\,Small\,Cuts}}

\newcommand{\Csub}{\mathcal{C}_{\hE}}
\newcommand{\Fsub}{\mathcal{F}_{\hE}}

\newcommand{\Flamsub}{\mathcal{F}_{\hE}^{\text{lam}}}

\newcommand\hJ{\widehat{J}}
\newcommand{\hL}{\widehat{\mathscr{L}}}

	\newcommand{\propgamma}{property\hbox{$\;(\gamma)$}}
	\newcommand{\propgamstar}{property\hbox{$\;(\gamma^{\star})$}}

\title{A $5$-Approximation Analysis for the Cover Small Cuts Problem}

\author{
\large
Miles Simmons\thanks{
        {\tt mjsimmons@uwaterloo.ca}.
        Department of Combinatorics \& Optimization, University of Waterloo, Canada.}
\and
Ishan Bansal\thanks{
        {\tt ib332@cornell.edu}.
	Amazon, Bellevue, WA, USA. This work is external and does not relate to the position at Amazon. }
\and
Joseph Cheriyan\thanks{
{\tt jcheriyan@uwaterloo.ca}.
        Department of Combinatorics \& Optimization, University of Waterloo, Canada.}
}

\begin{document}

\maketitle

\begin{center}
\textit{
The woods are lovely, dark and deep,\\
But I have promises to keep,\\
And miles to go before I sleep,\\ 
And miles to go before I sleep.\\
  \mbox{~} \qquad \qquad --Robert Frost}
\end{center}

\begin{abstract}{
In the $\ASC$ problem,
we are given a capacitated (undirected) graph $G=(V,E,u)$
and a threshold value $\lambda$,
as well as a set of links $L$ with end-nodes in $V$
and a non-negative cost for each link $\ell\in L$;
the goal is to find a minimum-cost set of links such that
each non-trivial cut of capacity less than $\lambda$ is covered by a link.

Bansal~et~al.\ \cite{BCGI24} showed that
the WGMV primal-dual algorithm, due to Williamson~et~al.\ \cite{WGMV95},
achieves approximation ratio~$16$ for the $\ASC$ problem;
their analysis uses the notion of a pliable family of sets that
satisfies a combinatorial property.
Later, Bansal \cite{B2024,B2025} and then Nutov \cite{N2025:mfcs,N2025:tightanalysis} proved that
the same algorithm achieves approximation ratio~$6$.
We show that the same algorithm achieves approximation ratio~$5$,
by using a stronger notion, namely,
a pliable family of sets that satisfies symmetry and structural submodularity.
}
\end{abstract}

\section{Introduction \label{sec:intro}}
{
Williamson~et~al.\ (WGMV) \cite{WGMV95} defined a family of sets $\F$ to be
\textit{uncrossable} if the following holds:
for any pair of sets $A,B\in\F$,
both $A\cap{B},A\cup{B}$ are in $\F$, or both $A-B,B-A$ are in $\F$.
They used this notion to design and analyse a primal-dual approximation
algorithm for the problem of covering an uncrossable family of sets,
and they proved an approximation ratio of two for their algorithm.
(The algorithmic problem is as follows:
The input consists of a family of sets $\F$ of a ground-set $V$,
a set of links $L$, where each link is an unordered pair of nodes of $V$,
and a non-negative cost for each link in $L$;
the goal is to find a minimum-cost set of links $L'$ such that
each set $S\in\F$ is covered by a link of $L'$,
i.e., $\delta_{L'}(S)\not=\emptyset,\;\forall{S}\in\F$.)
Recently, Bansal~et~al.\ \cite{BCGI24} defined a family of sets $\F$ to be \textit{pliable} if
the following holds:
for any pair of sets $A,B\in\F$,
at least two of the (four) sets $A\cap{B},A\cup{B},A-B,B-A$ are in $\F$.
Bansal~et~al.\ \cite{BCGI24} showed that the WGMV primal-dual algorithm \cite{WGMV95}
achieves an approximation ratio of~$16$ for
the problem of covering a pliable family of sets that satisfies \propgamma{}.
(We discuss \propgamma{} in the following section;
it is a combinatorial property, and the analysis of \cite{BCGI24} relies on it.)

In the $\ASC$ problem,
we are given a capacitated (undirected) graph $G=(V,E,u)$
and a threshold value $\lambda$,
as well as a set of links $L$ with end-nodes in $V$
and a non-negative cost for each link $\ell\in L$;
the goal is to find a minimum-cost set of links such that
each non-trivial cut of capacity less than $\lambda$ is covered by a link.

The $\ASC$ problem is a special case of the problem of
covering a pliable family of sets that satisfies \propgamma{},
hence, the WGMV primal-dual algorithm \cite{WGMV95}
achieves approximation ratio~$16$ for the former problem.
Later, Bansal \cite{B2024,B2025}, and then Nutov \cite{N2025:mfcs,N2025:tightanalysis} proved 
approximation ratio~$6$ for the same algorithm for 
a generalization of the $\ASC$ problem, namely,
the problem of covering a pliable family of sets that satisfies both \propgamma{} and
the (so-called) sparse crossing property (see section~\ref{sec:prelims}).
Moreover, Nutov \cite{N2025:mfcs} shows that the $6$-approximation analysis is tight
for the generalization.
There are other results by Nutov on related topics, see \cite{N2023,N2024:waoa,N2024}
(we mention that these papers are not directly relevant for the results in our paper).

Diestel~et~al.\ \cite{D2017,DEE2017,DEW2019,DO2019,EKT2021} introduced the notion of
abstract separation systems that satisfy a submodularity property, and
they call this structural submodularity.
One of their motivations was to
identify the few structural assumptions one has to make of a set
of objects called ‘separations’ in order to capture the essence of
tangles in graphs, and thereby make them applicable in wider contexts.

A pliable family of sets $\F$ satisfies \textit{structural submodularity} if the following holds:
for any pair of crossing sets $A,B\in\F$,
at least one of the sets $A\cap{B},A\cup{B}$ is in $\F$, and
at least one of the sets $A-B,B-A$ is in $\F$.
A family of sets $\F$ of a ground-set $V$ such that
$S\in\F$ iff $V-S\in\F$ is said to satisfy \textit{symmetry}.

We present an improved analysis of the WGMV primal-dual algorithm for the $\ASC$ problem
that achieves approximation ratio~$5$;
our analysis uses the notion of a pliable family of sets that
satisfies symmetry and structural submodularity.

We apply a result of Bansal \cite{B2024,B2025}, to improve the approximation ratio from $6$ to $5$. 
Bansal's result states the following:
The approximation ratio for the WGMV primal-dual algorithm applied
to the problem of covering a pliable family of sets $\F$ is $(3+\rho(\F))$, where $\rho(\F)$
denotes the so-called crossing density of $\F$.
Section~\ref{sec:prelims} discusses this result in detail.
We show that $\rho(\F)\leq2$ for a pliable family of sets $\F$ that
satisfies symmetry and structural submodularity, see section~\ref{sec:main}.

Recently, Nutov \cite{N2025:example} constructed an example of the $\ASC$ problem such that
the WGMV primal-dual algorithm finds a solution of cost $(5-\epsilon)\opt$,
where $\opt$ denotes the cost of an optimal (integral) solution and
$\epsilon>0$ is an arbitrarily small number.
This shows that our (upper~bound) analysis is tight, and concludes the quest
for the tight approximation ratio for the WGMV primal-dual algorithm applied to the $\ASC$ problem.

Simmons, in his thesis, \cite{S2025}, uses the notion of a strongly pliable family of sets.
This notion is the same as the notion of structural submodularity of Diestel~et~al.\ \cite{DEW2019},
and, in this paper, we use the term structural submodularity (rather than strongly pliable).
The results in this paper are based on a sub-chapter of the first author's thesis,
see \cite[Chapter~2.3.3]{S2025}.
}
\section{	Preliminaries \label{sec:prelims}}
{
For a positive integer $k$, we use $[k]$ to denote the set $\{1,2,\dots,k\}$.
A pair of subsets $A,B$ of $V$ (the ground-set) is said to \textit{cross} if each of the four sets
$A\cap{B}, V-(A\cup{B}), A-B, B-A$ is non-empty.

\begin{assumption}
When discussing a family of sets $\F$, we assume that the empty set and
the ground-set are not in $\F$. (Otherwise, the problem of covering $\F$ would be infeasible.)
\end{assumption}

For a graph $G=(V,E)$ and $S\subseteq{V}$, the cut $\delta_{E}(S)$
refers to the set of edges that have exactly one end-node in $S$;
$\delta_{E}(S)$ is called a {non-trivial} cut if $\emptyset\neq{S}\subsetneq{V}$.

For a family of sets $\F$ of a ground-set $V$ and any set of edges $\hE$
(where each edge in $\hE$ has both ends in $V$),
let $\Fsub$ denote the subfamily $\{S\in\F\,:\,\delta_{\hE}(S)=\emptyset\}$, i.e.,
$\Fsub$ consists of the sets of $\F$ that are not covered by $\hE$.
Moreover, let $\Csub$ denote the family of inclusion~minimal sets of $\Fsub$.
For a pair of sets $A,B\subseteq{V}$, if an edge $e$ is in one of
$\delta_{E}(A\cap{B}), \delta_{E}(A\cup{B}), \delta_{E}(A-{B}), \delta_{E}(B-{A})$,
then $e$ is in one of $\delta_{E}(A), \delta_{E}({B})$.
This implies the next fact; also see \cite[Facts~2.1, 2.2]{BCGI24}.

\begin{fact} \label{fact:restriction}
Let $\hE$ be any set of edges.
If a family of sets $\F$ is symmetric, then $\Fsub$ is symmetric.
If a family of sets $\F$ is pliable, then $\Fsub$ is pliable.
If a family of sets $\F$ satisfies structural submodularity, then so does $\Fsub$.
\end{fact}

A family of sets $\F$ satisfies \propgamma{} if
for any set of edges $\hE$ and sets $S_0, S_1 \in\Fsub$, $S_1\subsetneq S_0$,
if an inclusion~minimal set $C\in\Fsub$ crosses both $S_0,S_1$, then
the set $S_0 - (C \cup S_1)$ is either empty or it is in $\Fsub$, \cite{BCGI24}.

A family of sets $\F$ satisfies the \textit{sparse crossing property} if
for any set $S\in\F$,
the number of inclusion~minimal sets of $\F$ that cross $S$ is zero or one.

\begin{lemma} \label{lem:disjointcores-strongpliable}
Let $\F$ be a pliable family that satisfies symmetry and structural submodularity.
Then the inclusion~minimal sets of $\F$ are (pairwise) disjoint.
\end{lemma}
\begin{proof}
Let $C_1,C_2\in\F$ be inclusion~minimal sets.
If $C_1,C_2$ are not disjoint, then $\F$ would contain a proper subset of
$C_1$ or $C_2$, and this would contradict inclusion~minimality.
(In particular, if $C_1,C_2$ cross, then one of $C_1-C_2$ or $C_2-C_1$ would be in $\F$.)
\end{proof}

\begin{lemma} \label{lem:sparsecrossing-strongpliable}
Let $\F$ be a pliable family that satisfies symmetry and structural submodularity.
Then $\F$ satisfies the sparse crossing property.
\end{lemma}
\begin{proof}
We use a contradiction argument. Suppose that $S\in\F$ crosses two
inclusion~minimal sets $C_1,C_2\in\F$;
by the previous lemma, $C_1,C_2$ are disjoint.
Then, since proper subsets of $C_1$ or $C_2$ cannot be in $\F$,
both $S\cup{C_1}$ and $S\cup{C_2}$ would be in $\F$.
Note that $S\cup{C_1}$ and $S\cup{C_2}$ cross ($S\cup{C_1}\cup{C_2}=V$ is not possible, by symmetry).
Then $(S\cup{C_1}) - (S\cup{C_2}) = C_1-S$ or
     $(S\cup{C_2}) - (S\cup{C_1}) = C_2-S$ is in $\F$; this gives a contradiction.
\end{proof}

The next lemma states that \propgamma{} holds for any
pliable family that satisfies symmetry and structural submodularity.
We prove a stronger statement in section~\ref{sec:main}, see Proposition~\ref{propos:gamstar}.

\begin{lemma} \label{lem:propgamma-strongpliable}
Let $\F$ be a pliable family that satisfies symmetry and structural submodularity.
Then $\F$ satisfies \propgamma{}.
\end{lemma}
\begin{proof}
Let $\hE$ be a set of edges and
let the sets $S_0,S_1,C\in\Fsub$ satisfy the conditions of \propgamma{}.
Suppose that $S_0 - (C\cup{S_1})$ is nonempty; otherwise, we are done.
Since $C$ crosses $S_1$ (and no proper subset of $C$ is in $\Fsub$), $S_1\cup{C}$ is in $\Fsub$.
Observe that $S_0$ crosses $S_1\cup{C}$ (since $S_0,C$ cross, and $S_0 - (C\cup{S_1})$ is nonempty).
Then $S_0 - (S_1\cup{C}) \in\Fsub$, since $(S_1\cup{C})-S_0 = C-S_0 \not\in\Fsub$.
\end{proof}

\begin{lemma} \label{lem:cuts-strongpliable}
Let $G=(V,E,u)$ be a capacitated graph and let $\lambda\in\R$.
Then the family of non-trivial small cuts
$\{S\,:\, \emptyset\neq{S}\subsetneq{V},\: u(\delta_{E}(S)) < \lambda\}$ is
a pliable family that satisfies symmetry and structural submodularity.
\end{lemma}
\begin{proof}
Clearly, symmetry holds since $\delta_{E}(S)=\delta_{E}(V-S),\,\forall{S}\subseteq{V}$.
Structural submodularity follows from the submodular inequalities, see \cite[Chapter~1.2]{Frank:book};
in detail, for $A,B\subseteq{V}$, we have
$u(\delta_{E}(A)) + u(\delta_{E}(B)) \geq u(\delta_{E}(A\cap{B})) + u(\delta_{E}(A\cup{B}))$, and
$u(\delta_{E}(A)) + u(\delta_{E}(B)) \geq u(\delta_{E}(A-B)) + u(\delta_{E}(B-A))$.
\end{proof}

\subsection*{Crossing density and Bansal's theorem on the approximation ratio of the WGMV algorithm}
{
Bansal~\cite{B2024,B2025} introduced the notion of crossing density of a family of sets.

\begin{definition} The \textbf{crossing density} of a family of sets $\F$,
denoted $\rho(F)$, is the smallest positive integer such that
for any edge-set $\hE$ and a laminar subfamily $\Flamsub$ of $\Fsub$
the number of crossing pairs $(S,C)$, where $S\in\Flamsub$ and $C\in\Csub$,
	is $\leq \rho(F) \cdot |\Csub|$.
\end{definition}

\begin{theorem}\cite[Theorem~1]{B2025} \label{thm:Bcrossing}
The approximation ratio of the WGMV primal-dual algorithm applied to
the problem of covering a pliable family of sets $\F$ is $(3+\rho(F))$.
\end{theorem}

Let $\F$ be a pliable family of sets.
We briefly discuss the notion of witness sets, and related notions. 
Let $\hE$ be any set of edges.
Let $\hJ$ be an inclusion~minimal set of links that covers $\Csub$.
For a link $\ell\in\hJ$, a \textit{witness set} is a set $S_{\ell}\in\Fsub$
such that $\delta_{\hJ}(S_{\ell})=\{\ell\}$.
Thus, a witness set of a link $\ell\in\hJ$ is a set of $\Fsub$ that is covered by $\ell$ and
is not covered by any other link of $\hJ$.
By \cite[Lemma~3.2]{BCGI24}, there exists a laminar family of witness~sets of $\hJ$.
In other words, for each link $\ell\in\hJ$, there exists a witness set $S_{\ell}\in\Fsub$
such that $\{ S_{\ell}\,:\,\ell\in\hJ \}$ is a laminar family.

\begin{remark}
In order to apply Theorem~\ref{thm:Bcrossing} to a family of sets $\F$,
one needs to prove an upper~bound on the crossing density of $\F$, $\rho(F)$.
W.l.o.g., let $\hE$ be any set of edges.
Formally speaking, one would have to consider all possible laminar subfamilies
$\Flamsub$ of the subfamily $\Fsub$.
In fact, Bansal's proof of Theorem~\ref{thm:Bcrossing} focuses on
a laminar family of witness~sets $\hL$
(corresponding to an inclusion~minimal set of links that covers $\Csub$).
Thus, if one can show that 
the number of crossing pairs $(S,C)$ such that $S\in\hL$ and $C\in\Csub$
is $\leq \hat{\rho} \cdot |\Csub|$, then, by the analysis in Bansal's proof of Theorem~\ref{thm:Bcrossing},
it follows that
the approximation ratio of the WGMV primal-dual algorithm (for covering $\F$) is $\leq(3+\hat{\rho})$.
\end{remark}
}
}
\section{ The crossing density of a pliable family that satisfies
	symmetry and structural submodularity \label{sec:main}}
{
In this section, we focus on a pliable family of sets $\F$ that satisfies
symmetry and structural submodularity.
Our goal is to show that the crossing density of $\F$, $\rho(\F)$, is $\leq2$.
Then, Theorem~\ref{thm:Bcrossing} implies that the approximation ratio of the WGMV
primal-dual algorithm applied to the problem of covering $\F$ is $\leq5$.

We introduce an extension of \propgamma{} called \propgamstar{}.

\begin{definition}
A family of sets $\F$ satisfies \propgamstar{} if
for any set of edges $\hE$ and sets $S_0, S_1,\dots,S_k \in\Fsub$ ($k\geq1$),
where $S_1,\dots,S_k$ are pairwise disjoint (proper) subsets of $S_0$,
if an inclusion~minimal set $C\in\Csub\subseteq\Fsub$ crosses all of $S_0,S_1,\dots,S_k$, then
the set $S_0 - (S_1\cup\dots\cup{S_k}\cup C)$ is either empty or it is in $\Fsub$.
\end{definition}

{
\begin{figure}[htb]
    \centering
    \begin{tikzpicture}[scale=0.75]
        \node at (-1,2) {$S_0$};
        \draw[black, fill = blue, fill opacity = 0.2, line width=2pt,radius = 50pt] (0,0) circle;
        \draw[red,fill=white, line width=2pt,radius = 50pt] (2,0) circle;
        \draw[black, line width=2pt,radius = 50pt] (0,0) circle;
        \node at (3,2) {$C$};
        \draw[black,line width = 2pt,fill = white] (0.5,1) ellipse [x radius = 20pt, y radius = 10pt];
        \draw[black,line width = 2pt,fill = white] (0,0) ellipse [x radius = 20pt, y radius = 10pt];
        \draw[black,line width = 2pt] (0.5,-1) ellipse [x radius = 20pt, y radius = 10pt];
        \draw[red, line width=2pt,radius = 50pt] (2,0) circle;
        \node at (-0.1,0) {$S_j$};
        \node at (0.5,-1) {$S_{j+1}$};
        \node at (-1.0,-0.6) {$D_{j}$};
    \end{tikzpicture}
    \caption{Illustration of \propgamstar{}.
	The shaded region indicates $D_{j}   = S_0 - (S_1 \cup\dots\cup S_{j} \cup C)$.
	}
\end{figure}
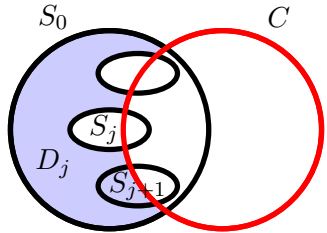
}

\begin{proposition}\label{propos:gamstar}
A pliable family of sets that satisfies symmetry and structural submodularity
necessarily satisfies \propgamstar{}.
\end{proposition}
\begin{proof}
Let $\hE$ be any set of edges, and let $S_0, S_1,\dots,S_k, C \in\Fsub$
satisfy the conditions of \propgamstar{}.
We may assume that $S_0 - (S_1\cup\dots\cup{S_k}\cup C)$ is nonempty.

We apply induction on $k$.

\begin{description}
\item{Induction basis, $k=1$:}
Then \propgamstar{} is the same as \propgamma{}, so
the result follows from Lemma~\ref{lem:propgamma-strongpliable}.

\item{Induction step:}
By the induction hypothesis,
the result holds for any $j$ pairwise disjoint (proper) subsets $S'_1,\dots,S'_j$ of $S_0$,
where $1\leq{j}<k$.
Let $S_0, S_1,\dots,S_{j+1}, C \in\Fsub$
satisfy the conditions of \propgamstar{}.
By the induction hypothesis,
the set $D_j = S_0 - (S_1\cup\dots\cup{S_j}\cup C)$ is nonempty and it is in $\Fsub$.
We claim that $D_j$ crosses the set $S_{j+1}$.
To verify this, note that
\begin{itemize}
\item[]
$V - (D_j\cup{S_{j+1}})$ is nonempty (since $S_0$ crosses $C$, $V-(S_0\cup{C})$ is nonempty),
\item[] $D_j\cap{S_{j+1}}={S_{j+1}}-C$ is nonempty (since ${S_{j+1}}$ crosses $C$),
\item[] $D_j - {S_{j+1}}$ is nonempty
	(by assumption, $S_0 - (S_1\cup\dots\cup{S_k}\cup C)$ is nonempty), and
\item[] ${S_{j+1}}-D_j = {S_{j+1}}\cap{C}$ is nonempty (since ${S_{j+1}}$ crosses $C$).
\end{itemize}
Since $D_j,{S_{j+1}}$ cross, $D_j-{S_{j+1}} = S_0 - (S_1\cup\dots\cup{S_j}\cup{S_{j+1}}\cup C)$
is in $\Fsub$, by structural submodularity; 
observe that ${S_{j+1}}-D_j = {S_{j+1}}\cap{C}$ is a proper subset of $C$ so it is not in $\Fsub$.
\end{description}
\end{proof}

Let $\hE$ be any set of edges.
Let $\hJ$ be an inclusion~minimal set of links that covers $\Csub$,
and let $\hL \subseteq \Fsub$ be a laminar family of witness~sets of $\hJ$,
see \cite[Lemma~3.2]{BCGI24}.
By the sparse crossing property of $\Fsub$ (see Fact~\ref{fact:restriction}),
each set in $\hL$ crosses zero or one sets in $\Csub$.

We define $\Ls$ to be the subfamily of sets in $\hL$ that cross at least one set of $\Csub$;
thus, $\Ls = \{ S\in\hL\,:\, \exists C\in\Csub ~\text{such that}~ S ~\text{crosses}~ C\}$.
Clearly, the sparse crossing property of $\Fsub$ implies that the number of crossing pairs
$(S,C)$ such that $S\in\hL$ and $C\in\Csub$ is $|\Ls|$;
note that the sets in $\hL - \Ls$ are not relevant for our analysis.

Now, our goal is to show that $|\Ls| \leq 2 \; |\Csub|$.
Before presenting our proof, let us informally sketch our analysis.
We view $\Ls\cup\{V\}$ as a rooted tree,
then we map each set $C\in\Csub$ to the smallest set $S\in\Ls\cup\{V\}$ that contains $C$
and we assign the color red to the tree-node corresponding to $S$;
finally, we show that each tree-node (other than the tree-node of $\{V\}$)
is red or it has a child tree-node that is red.
Clearly, the number of red tree-nodes is $\leq|\Csub|$, and it follows that $|\Ls|\leq2 \; |\Csub|$.
In other words, each set $C\in\Csub$ handles $\leq2$ tree-nodes, so we have $|\Ls|\leq2 \; |\Csub|$.

Next, we present our proof of $|\Ls| \leq 2 \; |\Csub|$.

{
Let $\Ls\cup\{V\}$ be the laminar family of witness sets (of $\hJ$)
together with the node-set $V$.
Let $\T$ be a rooted tree that represents $\Ls\cup\{V\}$;
for each set $S\in\Ls\cup\{V\}$, there is a node $v_S$ in $\T$,
and the node $v_V$ is taken to be the root of $\T$.
{Thus, $\T$ has an edge $\{v_Q,v_S\}$ such that $Q\supsetneq{S}$ iff $Q$ is the smallest set of $\Ls$
that properly contains the set $S$ of $\Ls$.}
Let $\psi$ be a mapping from $\Csub$ to $\Ls\cup\{V\}$ that
maps each set $C\in\Csub$ to the smallest set $S\in\Ls\cup\{V\}$ that contains it.
If a node $v_S$ of $\T$ has some set $C\in\Csub$ mapped to its associated set $S$, then we
assign the color red to $v_S$.
}
{
\begin{lemma}\label{lem:disjoint}
Let $S_0$ be a set in $\Ls$.
Let $C_0$ be the unique set of $\Csub$ that crosses $S_0$.
Suppose that $S_0$ contains a set $S_1\in\Ls$ such that 
$v_{S_1}$ is a child of $v_{S_0}$ in the tree $\T$ and, moreover,
$S_1, C_0$ are disjoint.
Then there is a set $C\in\Csub$ such that $\psi(C)=S_0$.
\end{lemma}
\begin{proof}
Since $S_1\in\Ls$, there exists a unique set $C_1\in\Csub$ that crosses $S_1$.
We claim that $C_1$ is a subset of $S_0$.
Otherwise, if $C_1-S_0$ is nonempty, then $C_1$ would cross $S_0$
(since $C_1\cap{S_0}$ is nonempty, and, moreover,
$C_0,C_1$ are disjoint, hence both $S_0-C_1$ and $V-(C_1\cup{S_0})$ are nonempty).
This gives a contradiction, since $C_0$ is the unique set of $\Csub$ that crosses $S_0$.
Observe that no proper subset of $S_0$ in $\Ls$ contains $C_1$, hence, $\psi(C_1)=S_0$.
\end{proof}

\begin{figure}[htb]
    \centering
     \begin{tikzpicture}[scale=0.75]
        \node at (-1,2) {$S_0$};
        \draw[black, line width=2pt,radius = 50pt] (0,0) circle;
        \draw[red,line width=2pt,radius = 50pt] (2.5,0) circle;
        \node at (3,2) {$C_0$};
        \draw[black, line width=2pt,radius = 15pt] (-0.5,-0.2) circle;
        \node at (-0.5,-0.3) {$S_1$};
        \draw[red, line width=2pt,radius = 15pt] (-0.5,0.5) circle;
        \node at (-0.5,0.6) {$C_1$};
    \end{tikzpicture}
    \caption{Illustration of Lemma~\ref{lem:disjoint}.}
\end{figure}
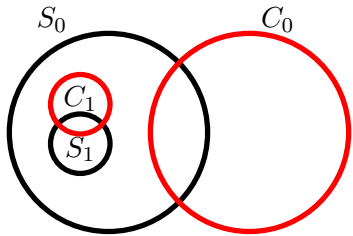

Let $S_0$ be any set in $\Ls$.
If there exist $k\geq1$ children of the node $v_{S_0}$ in the tree $\T$,
then we use the notation $v_{S'_1},\dots,v_{S'_k}$ to denote the children nodes
of $v_{S_0}$ in the tree $\T$, and
we denote the corresponding sets of $\Ls$ by $S'_1,\dots,S'_k$.
We call each of $S'_1,\dots,S'_k$ a child set of $S_0$.
Moreover, if a child node say $v_{S'_i}$ has $k_i\geq1$ children in the tree $\T$,
then we use the notation $v_{S''_{(i,1)}},\dots,v_{S''_{(i,k_i)}}$ to denote the children nodes
of $v_{S'_i}$ in the tree $\T$, and
we denote the corresponding sets of $\Ls$ by $S''_{(i,1)},\dots,S'_{(i,k_i)}$.
We call each of $S''_{(i,1)},\dots,S'_{(i,k_i)}$ a child set of $S'_i$.

\begin{lemma}\label{lem:empty-remainder}
Let $S_0$ be a set in $\Ls$.
Let $C_0$ be the unique set of $\Csub$ that crosses $S_0$.
Suppose that there is no set $C\in\Csub$ such that $\psi(C)=S_0$.
Then $S_0$ has $k\geq1$ children sets $S'_1,\dots,S'_k$ such that
	$C_0$ crosses each of $S'_1,\dots,S'_k$, and, moreover,
	$S_0 - (C_0 \cup S'_1 \cup \dots \cup S'_k)$ is the empty set.
\end{lemma}
\begin{proof}
Let $S_0\in\Ls$ and $C_0\in\Csub$ satisfy the statement of the lemma.
Our proof applies case analysis and we handle the cases using
Lemma~\ref{lem:disjoint} or \propgamstar{}.

\begin{description}[style=nextline]
{
\item[(Case 1) $S_0$ has no children sets:]
Since $S_0\in\Fsub$, it contains an inclusion~minimal set $C \in\Csub\subseteq\Fsub$;
clearly, $\psi(C)=S_0$.

\item[(Case 2) $S_0$ has a child set say $S'_i$ that is disjoint from $C_0$:]
Then, by Lemma~\ref{lem:disjoint}, for the unique set $C\in\Csub$ that crosses $S'_i$
we have $\psi(C)=S_0$.

\item[(Case 3) $S_0$ has $k\geq1$ children sets $S'_1,\dots,S'_k$,
$C_0$ crosses each of $S'_1,\dots,S'_k$, and
$S_0 - (C_0 \cup S'_1 \cup \dots \cup S'_k)$ is nonempty:]
By \propgamstar{}, the set $S_0 - (C_0 \cup S'_1 \cup \dots \cup S'_k)$ is in $\Fsub$,
hence, it contains an inclusion~minimal set $C \in\Csub\subseteq\Fsub$;
clearly, $\psi(C)=S_0$.

\item[] The remaining case is the one described in the conclusion of the lemma.
}
\end{description}
\end{proof}

Our analysis relies on the next lemma. This lemma states that for
each set $S\in\Ls$, there exists a set $C\in\Csub$ such that
$\psi(C)=S$ or $\psi(C)$ is a child of $S$ in the tree $\T$.

\begin{lemma}\label{lem:5approx}
For every set $S\in\Ls$ there exists a set $C\in\Csub$ such that
$\psi(C)=S$, or
$\psi(C)=S'$ where $S'\in\Ls$ is such that
$S'\subsetneq{S}$ and $S$ is the smallest set in $\Ls$ that contains $S'$.
\end{lemma}
{
\begin{proof}
Let $S_0$ be any set in $\Ls$.
Since $S_0\in\Ls$, there is a unique set $C_0\in\Csub$ that crosses $S_0$.

We start by applying Lemma~\ref{lem:empty-remainder} to $S_0$.
If there is a set $C\in\Csub$ such that $\psi(C)=S_0$, then we are done.

Otherwise, by Lemma~\ref{lem:empty-remainder},
     $S_0$ has $k\geq1$ children sets $S'_1,\dots,S'_k$ such that
	$C_0$ crosses each of $S'_1,\dots,S'_k$, and, moreover,
	$S_0 - (C_0 \cup S'_1 \cup \dots \cup S'_k)$ is the empty set,

Next, we apply Lemma~\ref{lem:empty-remainder} to the child set $S'_1$ of $S_0$
(since $k\geq1$, $S'_1$ exists).
If there is a set $C\in\Csub$ such that $\psi(C)=S'_1$, then we are done.

Otherwise, by Lemma~\ref{lem:empty-remainder},
     $S'_1$ has $k_1\geq1$ children sets $S''_{(1,1)},\dots,S''_{(1,k_1)}$ such that
	$C_0$ crosses each of $S''_{(1,1)},\dots,S''_{(1,k_1)}$, and
	$S'_1 - (C_0 \cup S''_{(1,1)} \cup \dots \cup S''_{(1,k_1)})$ is the empty set,

We handle this case via a contradiction argument.

Recall that every $S\in\Ls$ is the unique witness set of some link $\ell$ of $\hJ$
(recall that $\hJ$ is an inclusion~minimal link-set that covers $\Csub$).
Consider the link $\{x,y\}$ corresponding to the set $S'_1\in\Ls$.
Clearly, this link is present in $\delta_{\hJ}(S'_1)$ and this link is absent from
each of $\delta_{\hJ}(S_0), \delta_{\hJ}(S'_2),\dots,\delta_{\hJ}(S'_k)$;
moreover, this link is absent from each of
	$\delta_{\hJ}(S''_{(1,1)}),\dots,\delta_{\hJ}(S''_{(1,k_1)})$.
Thus, one end of the link, say $x$, is in $S'_1 - (S''_{(1,1)} \cup\dots\cup S''_{(1,k_1)})$,
and the other end of the link, $y$, is in $S_0 - (S'_1 \cup\dots\cup S'_k)$.
Recall that
	$S_0 -  (C_0 \cup S'_1 \cup \dots \cup S'_k)$ is the empty set, and
	$S'_1 - (C_0 \cup S''_{(1,1)} \cup \dots \cup S''_{(1,k_1)})$ is the empty set.
Hence, $x\in C_0$ and $y\in C_0$.
We have a contradiction, since the link $\{x,y\}$
does not cover any set of $\Csub$ (recall that the sets in $\Csub$ are pairwise disjoint);
therefore, $\hJ$ is not an inclusion~minimal cover of $\Csub$.
\end{proof}
}
}
{
   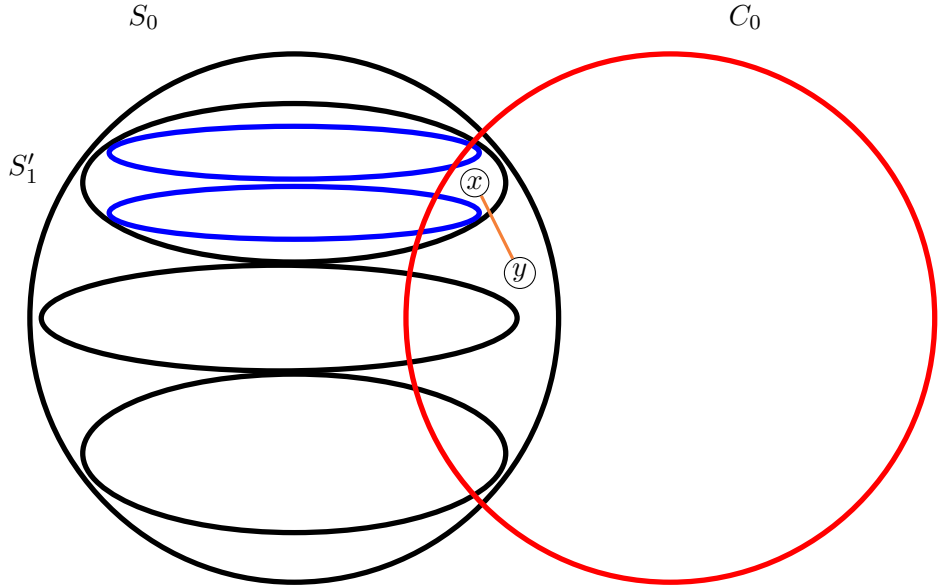
\begin{figure}[htb]
       \centering
       \begin{tikzpicture}[scale=2]
        \node at (-1,2) {$S_0$};
        \draw[black, line width=2pt,radius = 50pt] (0,0) circle;
        \node at (3,2) {$C_0$};
        \draw[black,line width=2pt] (-0.1,0) ellipse [x radius = 45pt, y radius = 10pt];
        \draw[black,line width=2pt] (0,-0.9) ellipse [x radius = 40pt, y radius = 15pt];
        \draw[black,line width=2pt] (0,0.9) ellipse [x radius = 40pt, y radius = 15pt];
        \draw[blue,line width=2pt] (0,1.1) ellipse [x radius = 35pt, y radius = 5pt];
        \draw[blue,line width=2pt] (0,0.7) ellipse [x radius = 35pt, y radius = 5pt];
        \draw[red,line width=2pt,radius = 50pt] (2.5,0) circle;
        \node at (-1.8,1) {$S'_1$};

         \node[draw,circle,black,fill=none,inner sep = 1pt] (a) at (1.2,0.9) {$x$};
         \node[draw,circle,black,fill=none,inner sep = 1pt] (b) at (1.5,0.3) {$y$};
         \draw[very thick, Orange] (a) edge (b);
    \end{tikzpicture}
       \caption{Illustration of the key case of Lemma~\ref{lem:5approx}.}
   \end{figure}
}

Summarizing the above discussion,
each set $C\in\Csub$ is mapped to a set $\psi(C)\in\Ls$;
moreover, for each set $S\in\Ls$,
there is a set $C\in\Csub$ such that either $\psi(C)=S$ or $\psi(C)=S'$
where $S'$ is a child set of $S$ in the tree $\T$.
The required inequality follows.

\begin{proposition} We have $\displaystyle |\Ls| \leq 2 \; |\Csub|. $
\end{proposition}

Then, by Theorem~\ref{thm:Bcrossing}, we get the following result.

\begin{theorem}
The WGMV primal-dual algorithm has approximation ratio $5$ when it is applied to
the problem of covering a pliable family that satisfies symmetry and structural submodularity.
\end{theorem}

}

\bibliographystyle{plainurl}
\bibliography{sbc-5approx-arxiv}

\end{document}